\documentclass[aps,pra,twocolumn,superscriptaddress,longbibliography]{revtex4-1} 

\usepackage{amsmath}
\usepackage{amsfonts}
\usepackage{amssymb}
\usepackage{graphicx}
\usepackage{color}
\usepackage{accents}
\usepackage[colorlinks=true, citecolor=cyan]{hyperref}
\usepackage{enumitem}
\usepackage{dsfont}
\usepackage{cleveref}
\usepackage{mathtools}

\newcommand{\Tr}{\mbox{Tr}}

\newcommand{\ket}[1]{|#1\rangle}
\newcommand{\bra}[1]{\langle #1|}

\usepackage{tikz}
\usetikzlibrary{quantikz}

\makeatletter 
\newsavebox{\@brx}
\newcommand{\llangle}[1][]{\savebox{\@brx}{\(\m@th{#1\langle}\)}%
  \mathopen{\copy\@brx\kern-0.5\wd\@brx\usebox{\@brx}}}
\newcommand{\rrangle}[1][]{\savebox{\@brx}{\(\m@th{#1\rangle}\)}%
  \mathclose{\copy\@brx\kern-0.5\wd\@brx\usebox{\@brx}}}
\makeatother

\newlength{\dhatheight} 

\newtheorem{theorem}{Theorem}[section]

\newtheorem{corollary}[theorem]{Corollary}
\newtheorem{claim}[theorem]{Claim}

\newenvironment{proof}[1][Proof.]{\begin{trivlist}
\item[\hskip \labelsep {\bfseries #1}]}{\end{trivlist}}

\newcommand{\qed}{\nobreak \ifvmode \relax \else
      \ifdim\lastskip<1.5em \hskip-\lastskip
      \hskip1.5em plus0em minus0.5em \fi \nobreak
      \vrule height0.75em width0.5em depth0.25em\fi}

\usepackage{amsmath}
\usepackage{enumitem}

\begin{document}

\title{Quantum Many-Body Scars in Dual-Unitary Circuits}
\author{Leonard Logari\'c}
\email[]{logaricl@tcd.ie}
\affiliation{Department of Physics, Trinity College Dublin, Dublin 2, Ireland}
\affiliation{Trinity Quantum Alliance, Unit 16, Trinity Technology and Enterprise Centre, Pearse Street, Dublin 2, D02 YN67, Ireland}
\author{Shane Dooley}
\email[]{dooleysh@gmail.com}
\affiliation{Department of Physics, Trinity College Dublin, Dublin 2, Ireland}
\affiliation{Trinity Quantum Alliance, Unit 16, Trinity Technology and Enterprise Centre, Pearse Street, Dublin 2, D02 YN67, Ireland}
\author{Silvia Pappalardi}
\affiliation{Institut f\"ur Theoretische Physik, Universit\"at zu K\"oln, Z\"ulpicher Straße 77, 50937 K\"oln, Germany}
\author{John Goold}
\affiliation{Department of Physics, Trinity College Dublin, Dublin 2, Ireland}
\affiliation{Trinity Quantum Alliance, Unit 16, Trinity Technology and Enterprise Centre, Pearse Street, Dublin 2, D02 YN67, Ireland}
\affiliation{Algorithmiq Limited, Kanavakatu 3C 00160 Helsinki, Finland}

\date{\today}

\begin{abstract}
Dual-unitary circuits are a class of quantum systems for which exact calculations of various quantities are possible, even for circuits that are nonintegrable. The array of known exact results paints a compelling picture of dual-unitary circuits as rapidly thermalizing systems. However, in this Letter, we present a method to construct dual-unitary circuits for which some simple initial states fail to thermalize, despite the circuits being ``maximally chaotic,'' ergodic and mixing. This is achieved by embedding quantum many-body scars in a circuit of arbitrary size and local Hilbert space dimension. We support our analytic results with numerical simulations showing the stark contrast in the rate of entanglement growth from an initial scar state compared to nonscar initial states. Our results are well suited to an experimental test, due to the compatibility of the circuit layout with the native structure of current digital quantum simulators.
\end{abstract}

\maketitle

{\bf \emph{Introduction.--}} 
Understanding how thermalization arises in closed quantum systems is a fundamental problem in many-body physics. Currently, our best understanding is based on the eigenstate thermalization hypothesis, which roughly states that thermalization occurs because the individual eigenstates of the unitary propagator appear thermal with respect to local observables \cite{DAl-16, Mor-18, Deu-91, Sre-94, Foi-19}. In this framework, failure to thermalize is due to the presence of nonthermal eigenstates. \emph{Strong ergodicity breaking} occurs in systems where a significant fraction of all the eigenstates are nonthermal, while \emph{weak ergodicity breaking} arises if the fraction of nonthermal eigenstates is exponentially small in system size \cite{Shi-17, Mou-22a, pap-2022, chan-23}. Strong ergodicity breaking is generally observed in systems with an extensive number of local conserved quantities, such as integrable \cite{Cal-16} or many-body localized (MBL) systems \cite{Nan-15a, Aba-19}, while weak ergodicity breaking is usually observed in systems with nonthermal eigenstates known as \emph{quantum many-body scars} (QMBS) \cite{Tur-18a, Ser-21, Mou-22a, pap-2022, chan-23, Des-22a, Doo-21a, Doo-23a}.

Because of the complex nature of interacting many-body dynamics, exact results are notoriously hard to come by. Significant progress has been made in recent years, by studying a special class of models, called \emph{dual-unitary (DU) circuits}. 
These are a class of quantum circuits constructed as a brickwork pattern of two-qudit gates, which are unitary in both temporal and spatial directions. This special property enables the exact calculation of some system properties that would ordinarily be prohibitively hard to calculate 
\cite{Aki-16, Ber-18a, Ber-19a, Cla-20a, Gop-19a, Ber-19b, Ber-20a, Pir-20a, Ber-21a, Ara-21a, Fri-21a, Cla-21a, Ler-21, Cher-21, Zho-22a, Suz-22, Cla-22a, Poz-22}.

For example, exact calculation of the spectral form factor has shown that interacting DU circuits are ``maximally chaotic,'' in the sense that the spectral form factor agrees with the predictions of random matrix theory at all timescales \cite{Ber-18a, Ber-21a}. Exact results also indicate that DU circuits are fast scramblers of quantum information, since two-time correlation functions and out-of-time correlators spread at their maximal possible velocities \cite{Ber-19a, Cla-20a}. In a similar spirit, from a certain class of solvable initial states it has been shown that entanglement growth occurs at the maximal rate \cite{Ber-19b, Pir-20a, Zho-22a} and that any finite subsystem thermalizes to its maximally mixed (i.e., infinite temperature) reduced density matrix in a short finite time \cite{Pir-20a}. Moreover, the exact results on two-time correlation functions allow a rigorous classification of DU circuits in terms of their ergodic and mixing properties \cite{Ber-19a, Cla-21a}.

The above exact properties seem to indicate that generic DU circuits are rapidly thermalizing systems \cite{Sin-21a}. This is further supported by the observation that it is impossible to induce MBL through disorder in DU circuits \cite{Ber-18a, Ber-21a}, leaving integrability as the only known mechanism of strong ergodicity breaking in DU circuits. However, to the best of our knowledge, so far there has been no demonstration of weak ergodicity breaking in DU circuits, nor any arguments against it as in the case of MBL.

In this Letter, we show that weak ergodicity breaking is indeed possible in DU circuits. Using a projector-embedding approach, initially proposed for continuous-time dynamics \cite{Shi-17}, we provide an explicit construction to insert QMBS in DU circuits of arbitrary system size and arbitrary local Hilbert space dimension. We demonstrate our construction with examples that embed a single, a few, or exponentially many QMBS in this class of circuits. Our construction shows that provably maximally chaotic, ergodic and mixing systems can support QMBS. Despite the rapid scrambling properties of such systems, if the system is initialized in the QMBS subspace then all quantum information remains localized in the subspace. Our analytical results are supported by numerical calculations for the entanglement growth, showing a striking difference in the dynamics between scar and nonscar initial states. Our work paves the way for future theoretical as well as experimental investigations of weak ergodicity breaking in quantum circuits, where dual unitarity can be leveraged \cite{Ste-22a, Koh-23}. Because of the native structure of current digital quantum simulators, some of the proposed models can be adapted directly to current devices. \\

{\bf \emph{Dual-unitary Circuits.--}} The basic building blocks of the circuits considered in this work are unitary operators $\hat{U}$, acting on two qudits of arbitrary local Hilbert space dimension $d$. The \emph{dual} $\tilde{U}$ of a unitary is defined by a reordering of the subsystem indices $\bra{k} \otimes \bra{l} \tilde{U} \ket{i} \otimes \ket{j} = \bra{j} \otimes \bra{l} \hat{U} \ket{i} \otimes \ket{k}$ \cite{Ber-19a} and can be physically interpreted as an exchange of the spatial and temporal dimensions of $\hat{U}$. A gate $\hat{U}$ is DU if and only if both $\hat{U}$ and its dual $\tilde{U}$ are unitary. Any two-qudit gate of the form
\begin{equation}\hat{U}^\text{DU,1}= (\hat{u}^+\otimes\hat{u}^-)\hat{S}\hat{V}(\hat{v}^-\otimes\hat{v}^+),\label{eq:new_DU_1}\end{equation} is dual unitary \cite{Poz-22,Pro-21a}. (See Supplemental Material \cite{SupMat}, Sec. I, for further details.)
Here, $\hat{S}\ket{i}\otimes\ket{j}=\ket{j}\otimes\ket{i}$ is the SWAP operator, $\hat{u}^\pm$ and $\hat{v}^\pm$ are arbitrary single-qudit unitaries, and \begin{equation}\hat{V}=\exp\{i\sum_{j=0}^{d-1}\hat{h}^{(j)}\otimes\ket{j}\bra{j}\},\label{eq:V}\end{equation} is an entangling gate with arbitrary single-qudit Hermitian operators $\hat{h}^{(j)}$. We rewrite the nonentangling unitaries as
\begin{align} \label{eq: single_qudit_f}
\hat{u}^+\otimes\hat{u}^-=\exp\{i (\hat{f}^+\otimes\hat{\mathbb{I}}+\hat{\mathbb{I}}\otimes\hat{f}^-)\},  
\\ \label{eq: single_qudit_g}
\hat{v}^-\otimes\hat{v}^+=\exp\{i (\hat{g}^-\otimes\hat{\mathbb{I}}+\hat{\mathbb{I}}\otimes\hat{g}^+) \} , 
\end{align} 
in terms of the single-qudit Hermitian operators $\hat{f}^\pm$ and $\hat{g}^\pm$. With this parametrization, a dual-unitary gate is specified by $\{ \hat{f}^\pm$, $\hat{g}^\pm$, $\hat{h}^{(j)} \}$. Note that, for any $\hat{U}^\text{DU,1}$, the gate \begin{equation} \hat{U}^\text{DU,2} = \hat{S}\,  \hat{U}^\text{DU,1} \hat{S} \label{eq:new_DU_2} \end{equation} is also DU. Generally, this expression is not expressible in the form of Eq. \eqref{eq:new_DU_1}, thus giving a distinct parametrization.

A DU circuit is a quantum circuit in a ``brickwork'' geometry, in which all of the two-qudit gates are DU. We consider an even number $N$ of qudits, where each qudit has a local Hilbert space $\mathbb{C}^{d}$. The qudit sites are labeled by $n = 0, 1, 2, ..., N - 1$, and we impose periodic boundary conditions $n \equiv n + N$. The basic building blocks of the dynamics are local DU gates $\hat{U}_{n,n+1}$. A single time step is implemented by a Floquet unitary operator $\hat{\mathbb{U}} = \hat{\mathbb{U}}_o \hat{\mathbb{U}}_e$ which is a layer of DU gates across even-odd bonds $\hat{\mathbb{U}}_e = \bigotimes_{j=0}^{N/2 - 1} \hat{U}_{2j,2j+1}$ and a layer of DU gates across odd-even bonds $\hat{\mathbb{U}}_o = \bigotimes_{j=1}^{N/2} \hat{U}_{2j-1,2j}$. We set all gates in the even layer to be identical to each other and in the form of Eq. \eqref{eq:new_DU_1}, $\hat{U}_{0,1} = \hat{U}_{2,3} = \cdots = \hat{U}^\text{DU,1}$, and all gates in the odd layer to be identical to each other and in the form of Eq. \eqref{eq:new_DU_2}, i.e., $\hat{U}_{1,2} = \hat{U}_{3,4} = \hdots = \hat{U}^{DU,2} = \hat{S} \hat{U}^\text{DU,1} \hat{S}$. However, this two-site translation invariance is not essential for our results, which are also valid if the gates vary from site to site. Evolution for $t\in\mathbb{Z}$ time steps is generated by powers of the Floquet operator  $\hat{\mathbb{U}}^t$. \\

{\bf \emph{Embedding QMBS in DU circuits.--}} 
In order to construct DU circuits with QMBS, we employ a projector embedding method, similar to the construction initially proposed by Shiraishi and Mori for continuous-time dynamics in Ref. \cite{Shi-17}. The essential idea is to use projectors in the generators of the unitary gates in Eqs. \eqref{eq:V}--~\eqref{eq: single_qudit_g} so that a chosen set of target states evolves by an elementary brickwork circuit of two-qudit swap gates, while states outside the target set evolve by more complicated dynamics.

Let $\hat{P}_{n, n+1}$ denote two-qudit projectors acting on neighboring sites $n$ and $n+1$ and define the extended projectors acting on the total system of $N$ qudits as $\hat{\mathbb{P}}_{n, n+1} \equiv \hat{\mathbb{I}}_{0, n-1} \otimes \hat{P}_{n, n+1} \otimes \hat{\mathbb{I}}_{n+2, N-1}$, where $\hat{\mathbb{I}}_{i, j}$ is the identity acting on all qudits in the range $i, i+1, ..., j-1, j$. The common kernel of all projectors is the set of states that are simultaneously annihilated by all projectors:
\begin{equation} \mathcal{K} = \{ \ket{\psi} : \hat{\mathbb{P}}_{n,n+1} \ket{\psi} = 0, \forall \: n \} . \end{equation} Our target set of states $\mathcal{T}$, which we wish to embed in our circuit as QMBS, is the subset of $\mathcal{K}$ that is invariant under the action of even and odd layers of SWAP gates: \begin{equation} \mathcal{T} = \{ \ket{\psi} : \ket{\psi}\in\mathcal{K}, \: \hat{\mathbb{S}}_{e}\ket{\psi}\in\mathcal{K}, \: \hat{\mathbb{S}}_{o}\ket{\psi}\in\mathcal{K} \},\end{equation} 
where $\hat{\mathbb{S}}_{e} = \bigotimes_{j=0}^{N/2-1}\hat{S}_{2j,2j+1}$ is the even layer of SWAP gates and $\hat{\mathbb{S}}_{o} = \bigotimes_{j=1}^{N/2}\hat{S}_{2j-1,2j}$ is the odd layer. 

We now outline a construction to embed $\mathcal{T}$ as a set of scars in our DU circuit. To do this, we impose three conditions on the generators $\{\hat{f}^\pm, \hat{g}^\pm, \hat{h}^{(j)} \}$ that define the DU gates in Eqs. \eqref{eq:new_DU_1}--\eqref{eq:new_DU_2} and on the projectors $\{ \hat{P}_{n,n+1} \}$ that define the target subspace: \begin{equation}  \hat{P}_{n,n+1} (\hat{f}_n^+ \otimes \hat{\mathbb{I}} + \hat{\mathbb{I}} \otimes \hat{f}_{n+1}^-) \hat{P}_{n,n+1} = \hat{f}_n^+ \otimes \hat{\mathbb{I}} + \hat{\mathbb{I}} \otimes \hat{f}_{n+1}^- , \label{eq:C1} \end{equation}
\begin{equation}  \hat{P}_{n,n+1} (\hat{g}_n^- \otimes \hat{\mathbb{I}} + \hat{\mathbb{I}} \otimes \hat{g}_{n+1}^+) \hat{P}_{n,n+1} = \hat{g}_n^- \otimes \hat{\mathbb{I}} + \hat{\mathbb{I}} \otimes \hat{g}_{n+1}^+ , \label{eq:C2} \end{equation}
\begin{equation} \hat{P}_{n,n+1} \Big( \sum_{j=0}^{d-1} \hat{h}_n^{(j)} \otimes \ket{j}\bra{j}_{n+1} \Big) \hat{P}_{n,n+1}  = \sum_{j=0}^{d-1} \hat{h}_n^{(j)} \otimes \ket{j}\bra{j}_{n+1} . \label{eq:C3} 
\end{equation}
Clearly, two-qudit gates satisfying conditions \eqref{eq:C1}--\eqref{eq:C3} are still of the DU form in Eq. \eqref{eq:new_DU_1}. Conditions \eqref{eq:C1}--\eqref{eq:C3} ensure that the unitaries $\hat{u}_n^+ \otimes \hat{u}_{n+1}^-$, $\hat{v}_n^+ \otimes \hat{v}_{n+1}^-$, and $\hat{V}_{n,n+1}$ act trivially on all states in the kernel $\mathcal{K}$ of the projectors $\hat{P}_{n,n+1}$. However, the dynamics are not necessarily closed in the subspace $\mathcal{K}$. This is because the parametrization of Eq. \eqref{eq:new_DU_1} includes the SWAP operators $\hat{S}_{n,n+1}$ which can potentially take states out of $\mathcal{K}$. So, the QMBS subspace is the subspace $\mathcal{T} \subset \mathcal{K}$ that is invariant under even $\hat{\mathbb{S}}_e$ and odd $\hat{\mathbb{S}}_o$ layers of SWAP operators. For any initial state $\ket{\psi} \in \mathcal{T}$ the circuit will then act as $\hat{\mathbb{U}}\ket{\psi} = \hat{\mathbb{S}} \ket{\psi}$, where $\hat{\mathbb{S}} = \hat{\mathbb{S}}_o \hat{\mathbb{S}}_e$, corresponding to integrable dynamics in the target subspace $\mathcal{T}$, even if the overall circuit $\hat{\mathbb{U}}$ is nonintegrable and results in complicated dynamics for initial states that are not in $\mathcal{T}$. The action of $\hat{\mathbb{U}}$ in $\mathcal{T}$ can, thus, be seen as a permutation of the qudits and many dynamical properties can, therefore, be studied exactly.

This prescription embeds QMBS in a DU circuit, which allows us to use all known exact results for these models. One of the key results for DU circuits is the rigorous classification in terms of their ergodic and mixing properties, based on the spectrum of a map $\mathcal M$ which determines the dynamical correlation functions \cite{Ber-19a, Cla-21a}. We can show that our construction leads to QMBS in provably ergodic and mixing many-body systems, since all dynamical correlation functions decay to their infinite temperature values at long times. Further details are provided in Supplemental Material \cite{SupMat}, Sec. II. \\

\begin{figure}
    \includegraphics[width = 0.48\textwidth]{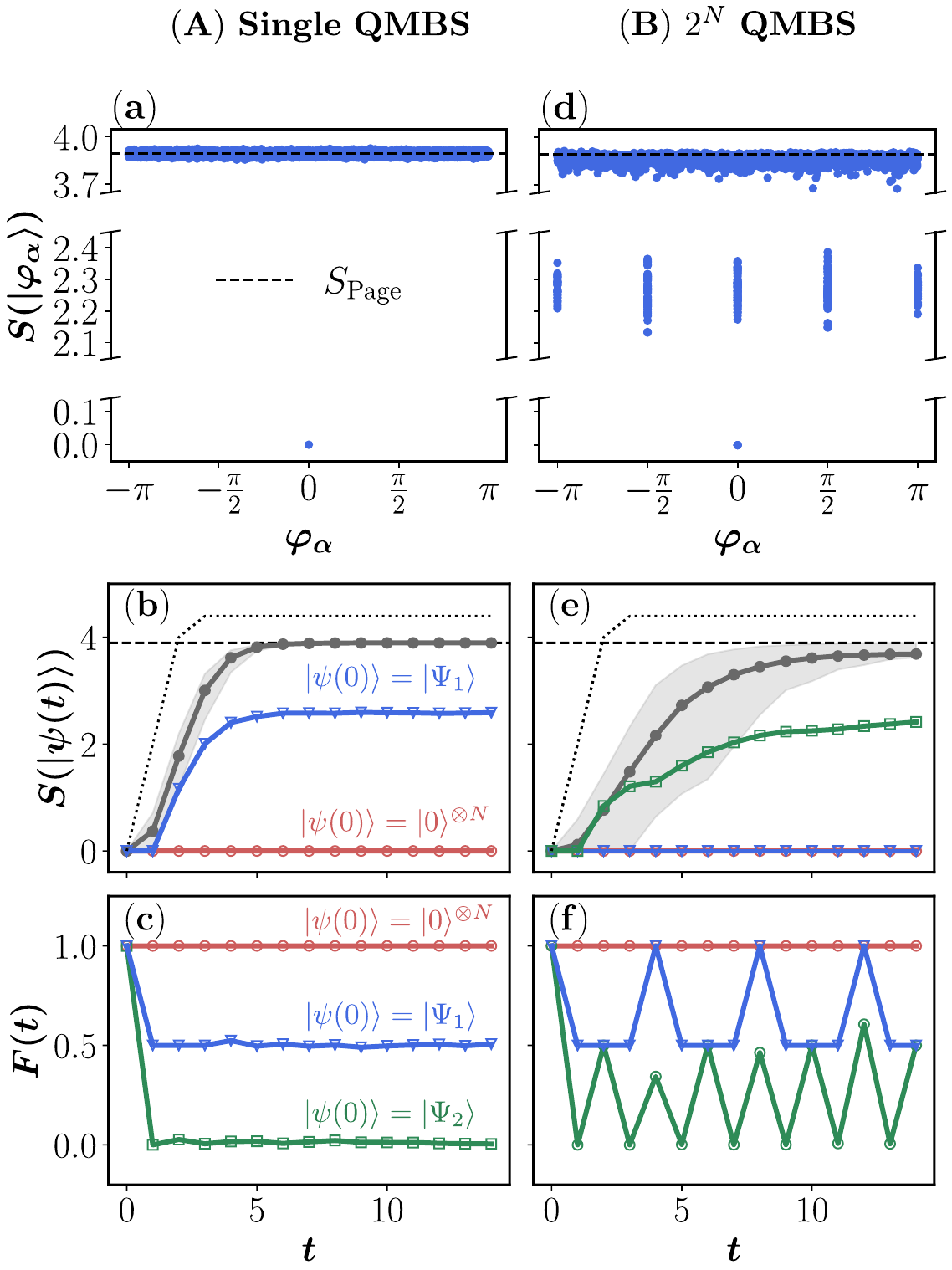}
    \caption{Embedding a single QMBS (A) and an exponential number of QMBS (B) in a DU circuit. The top row (a),(d) shows the bipartite entanglement of eigenstates $\ket{\varphi_\alpha}$ of the Floquet unitary $\hat{\mathbb{U}}$ in each case. The second row (b),(e) shows the growth of bipartite entanglement, starting from different separable states. The gray line is the average entanglement growth for 100 random product states of the form $\ket{\psi(0)} = \ket{i_0,i_1,\hdots,i_{N-1}}$, $i_n \in \{0,1,\hdots,d-1 \}$, while the surrounding gray region shows the 5-95 percentile range. The dashed black line in (a), (b), (d), and (e) is the Page entropy, $S_\text{Page} = [N \text{ln}(d) - 1]/2$. The dotted line in (b) and (d) shows the maximum possible rate of entanglement growth, saturating at the maximum value $S_{\text{max}} = N \text{ln}(d)/2$. ``Solvable'' initial states have this maximal entanglement growth in the thermodynamic limit \cite{Ber-19b,Pir-20a,Fol-23}. The bottom row (c), (f) shows the overlap with the initial state $F(t)=|\langle \psi(0) | \psi(t)\rangle|$. The blue line in (b), (c), (e), and (f) shows the results for $\ket{\Psi_1} = \ket{0}^{N-1}\otimes [(\ket{0} + \ket{d-1})/\sqrt{2}]$, and the green line for $\ket{\Psi_2} = \ket{0,0,d-1,d-1}^{\otimes N/4-1}\otimes\ket{0,0,d-1}\otimes (\ket{d-1} + \ket{1})/\sqrt{2}$.  [Parameters: $N=8$ and $d=3$]}
    \label{fig:QMBS_in_DUC}
\end{figure}

{\bf \emph{Example A: Single QMBS.--}} 
We first illustrate our construction with the simplest possible example, for which the set of projectors is \begin{equation} \hat{P}_{n,n+1} = \hat{\mathbb{I}}_{n,n+1} - \ket{0}\bra{0}_n \otimes \ket{0}\bra{0}_{n+1} , \label{eq:P_1} \end{equation} for all $n = 0,1,\hdots,N-1$. The common kernel $\mathcal{K}$ of all projectors consists of a single state $\ket{0}^{\otimes N}$, which is invariant under the action of any SWAP gates so that the target space is the single state $\mathcal{T} = \{ \ket{0}^{\otimes N} \}$. 

We choose $\{ \hat{f}^\pm, \hat{g}^\pm, \hat{h}^{(j)} \}$ to be random Hermitian matrices, apart from some rows and columns $\bra{i}\hat{f}^\pm \ket{0}=\bra{i}\hat{g}^\pm\ket{0}=\bra{i}\hat{h}^{(0)}\ket{0}=0$, $i\in\{0,1,...,d-1\}$, that are set to zero to ensure that conditions \eqref{eq:C1}--\eqref{eq:C3} are satisfied (see Supplemental Material \cite{SupMat}, Sec. III., for further details). This choice gives a generic DU gate which we expect to give a rapidly thermalizing circuit. Moreover, we can prove that for this example our circuit is ergodic and mixing \cite{Ber-19a, Cla-21a} (see Supplemental Material \cite{SupMat}, Sec. II). However, by our construction the target state $\ket{0}^{\otimes N}$ will be a nonthermal eigenstate of the circuit $\hat{\mathbb{U}}\ket{0}^{\otimes N} = \ket{0}^{\otimes N}$.

Although this construction embeds a QMBS for any system size $N$, we demonstrate our results with numerical simulations on finite size systems. In Fig. \ref{fig:QMBS_in_DUC}(a), we plot the half-system bipartite entanglement entropy $S(\ket{\varphi_{\alpha}}) = -\text{Tr}[\hat{\rho}_{\alpha} \text{log}(\hat{\rho}_\alpha)]$ of the eigenstates $\ket{\varphi_{\alpha}}$ of the Floquet unitary $\hat{\mathbb{U}}$. Here $\hat{\rho}_\alpha = \Tr_{0,(N/2)-1} \ket{\varphi_\alpha}\bra{\varphi_\alpha}$ is the reduced density matrix of the eigenstate obtained by tracing out the first $N / 2$ qudits of the system. We see the single QMBS $\ket{0}^{\otimes N}$ with zero entanglement, separated from the rest of the spectrum. All the other eigenstates are concentrated around the Page entropy $S_\text{Page}$, which is the expected value of entanglement entropy for a random pure state \cite{Pag-93}. Further evidence of the QMBS can be seen in the dynamics. In Fig. \ref{fig:QMBS_in_DUC}(b), we show that, for the initial QMBS state $\ket{\psi(0)} = \ket{0}^{\otimes N}$, the bipartite entanglement $S[\ket{\psi(t)}]$ does not grow in time, and in Fig. \ref{fig:QMBS_in_DUC}(c) that the fidelity $F(t) = |\bra{\psi(0)} \hat{\mathbb{U}}^t \ket{\psi(0)}|$ remains constant at $F=1$. By contrast, a typical random initial product state of the form $\ket{\psi(0)} = \ket{\vec{i}} \equiv \ket{i_0,i_1,\hdots,i_{N-1}}$, with $i_n$ chosen uniformly at random from the set $\{0,1,\hdots,d-1 \}$, will give rapid growth in entanglement that saturates at the Page value $S_\text{Page}$ and will give a fidelity that decays rapidly to zero. 

We further devise an initial state with intermediate slow dynamics. We consider the product state $\ket{\psi(0)} = \ket{\Psi_1} \equiv \ket{0}^{N-1}\otimes [(\ket{0} + \ket{d-1})/\sqrt{2}]$. This initial state is partly overlapping with the QMBS $\ket{0}^{\otimes N}$ and is partly overlapping with the thermalizing state $\ket{0}^{\otimes (N-1)}\otimes\ket{d-1}$. Correspondingly, the entanglement entropy saturates to a value intermediate between zero and the thermal value, while the fidelity approaches $F \approx 1/2$ after $t \sim 1$ time step.

This construction can be extended to embed $k < d$ QMBS of the form $\{ \ket{i}^{\otimes N} \}_{i=0}^{k-1}$ into a DU circuit with local Hilbert space dimension $d$. The main adjustment is, instead of Eq. \eqref{eq:P_1}, to use the set of projectors $\hat{P}_{n,n+1} = \hat{\mathbb{I}}_{n,n+1} - \sum_{i=0, }^{k-1}\ket{i,i}\bra{i,i}_{n,n+1}$. We provide numerical demonstrations for the $k=2$ case in Supplemental Material \cite{SupMat}, Sec. V, as well as a discussion on how the circuit can be modified to break the eigenphase degeneracy between the QMBS \cite{SupMat}, Sec IV. We also show that the QMBS in the Floquet unitary $\hat{\mathbb{U}}$ have counterpart \emph{dual} QMBS in the dual Floquet unitary $\tilde{\mathbb{U}}$ \cite{SupMat}, Sec. VI. \\

{\bf \emph{Example B: Exponentially many QMBS.--}} 
A more complex example is obtained by using the projectors: \begin{eqnarray} \hat{P}_{n,n+1} = \hat{\mathbb{I}}_{n,n+1}  &-& \ket{0}\bra{0}_n \otimes \ket{0}\bra{0}_{n+1} \label{eq:P_exp} \\ &-& \ket{0}\bra{d-1}_n \otimes \ket{d-1}\bra{0}_{n+1} \nonumber \\ &-& \ket{d-1}\bra{0}_n \otimes \ket{0}\bra{d-1}_{n+1} \nonumber \\ &-& \ket{d-1}\bra{d-1}_n \otimes \ket{d-1}\bra{d-1}_{n+1} , \nonumber \end{eqnarray} for all $n$. We choose the DU gate generators $\{ \hat{f}^\pm, \hat{g}^\pm, \hat{h}^{(j)} \}$ to be random, apart from a few rows and columns that are set to zero to ensure that conditions \eqref{eq:C1}--\eqref{eq:C3} are satisfied [$\bra{i}\hat{f}^\pm\ket{j}=\bra{i}\hat{g}^\pm\ket{j}=\bra{i}\hat{h}^{(j)}\ket{j}=0$, $j\in\{0,d-1\}$, $i\in\{0,1,\hdots,d-1\}$].

The common kernel $\mathcal{K}$ of the projectors in Eq. \eqref{eq:P_exp} is the $2^N$-dimensional space spanned by the set of all $N$-qudit product states of the form $\ket{\vec{i}} = \ket{i_0,i_1,\hdots,i_{N-1}}$ with $i_n \in \{ 0, d-1 \}$. We note that this space, though exponentially large in the system size $N$, is still an exponentially small fraction $(2/d)^N$of the full Hilbert space when $d>2$. Therefore, the fraction of nonthermal eigenstates constitutes a measure-zero set in the thermodynamic limit and, thus, leads to weak ergodicity breaking, consistent with the definition in Refs. \cite{Ser-21, Mou-22a}. It is also invariant under the action of $\hat{\mathbb{S}}_{e/o}$, so that our target space is $\mathcal{T} = \mathcal{K}$. 
However, this does not mean that the product states $\ket{\vec{i}}$, $i_n \in \{ 0, d-1 \}$, are QMBS since, in general, $\hat{\mathbb{U}}\ket{\vec{i}} = \hat{\mathbb{S}}\ket{\vec{i}} \neq \ket{\vec{i}}$. Instead, the QMBS are found by diagonalizing the SWAP circuit $\hat{\mathbb{S}}$ in the subspace spanned by $\{ \ket{\vec{i}} \}_{i_n \in \{ 0, d-1 \} }$.

In Fig. \ref{fig:QMBS_in_DUC}(d), we plot the entanglement entropy of the eigenstates of $\hat{\mathbb{U}}$. The bulk of the eigenstates have an entropy close to the Page value $S_\text{Page}$. However, the $2^N$ states in $\mathcal{T}$ have a lower entanglement. There are, in fact, four QMBS that have zero entanglement. These are $\ket{0}^{\otimes N}$, $\ket{d-1}^{\otimes N}$, $\ket{0,d-1}^{\otimes N/2}$, and $\ket{d-1,0}^{\otimes N/2}$, which are eigenstates of the SWAP circuit $\hat{\mathbb{S}}$. The remaining $2^N - 4$ QMBS have nonzero entanglement but are still clearly separated from the bulk of the spectrum. In Supplemental Material \cite{SupMat}, Sec. VII, we prove that the entanglement entropy of these QMBS can be bounded by $S \leq \log(N/2)$, a subvolume law scaling \cite{Sur-PC}.

The presence of exponentially many QMBS can also be observed in the dynamics, as shown in the entanglement entropy growth and fidelity dynamics in Figs. \ref{fig:QMBS_in_DUC}(e) and \ref{fig:QMBS_in_DUC}(f). When the system is initialized in a typical random product state (grey in the plot), it rapidly thermalizes to a highly entangled state.  Conversely, for an initial product state in the QMBS subspace $\mathcal{T}$, the entanglement growth is completely suppressed, since the SWAP circuit $\hat{\mathbb{S}}$ can map only product states to product states. For instance, the initial states $\ket{\psi(0)} = |0\rangle^{\otimes N}$ and $\ket{\psi(0)} = \ket{\Psi_1} \equiv \ket{0}^{N-1}\otimes (\ket{0} + \ket{d-1})/\sqrt{2}$ [red and blue lines, respectively, in Fig. \ref{fig:QMBS_in_DUC}(e)] show zero entanglement growth as they are both in the QMBS subspace $\mathcal{T}$.

However, the frozen entanglement from initial product states in $\mathcal{T}$ does not capture all the information about the dynamics in the QMBS subspace nor the variety of dynamical behaviors that are possible. For instance, evolution from $\ket{\psi(0)}=\ket{\Psi_1} = \ket{0}^{N-1}\otimes (\ket{0} + \ket{d-1})/\sqrt{2}$ is characterized by an oscillating fidelity $F=|\langle\psi(0)|\psi(t)\rangle|$. This is due to the fact that $\ket{\Psi_1}$, although it is in the QMBS subspace, is a superposition of multiple QMBS. The $\ket{0}^{\otimes N}$ component of $\ket{\Psi_1}$ is a QMBS and does not decay, while the $\ket{0}^{\otimes N-1}\otimes\ket{d-1}$ component is a superposition of QMBS and leads to fidelity oscillations with a period $T=N/2$. Alternatively, by initializing in $\ket{\Psi_2} = \ket{0,0,d-1,d-1}^{\otimes N/4-1}\otimes\ket{0,0,d-1}\otimes (\ket{d-1} + \ket{1})/\sqrt{2}$ (green in Fig. \ref{fig:QMBS_in_DUC}) the system undergoes more rapid oscillatory dynamics. This initial state is an equal superposition of a state $\ket{0,0,d-1,d-1}^{\otimes N/4}$ which is in the QMBS subspace and one that is not $\ket{0,0,d-1,d-1}^{\otimes N/4 - 1}\otimes\ket{0,0,d-1,1}$. While the former leads to revivals in fidelity with period $T=2$, the latter contribution rapidly decays to zero. 

In Supplemental Material \cite{SupMat}, Sec. VIII, we report further numerical results, including the dynamics of local observables, illustrating the presence of QMBS. \\ 

{\bf \emph{Discussion.--}}
Dual-unitary circuits are a paradigmatic model for investigations into many-body phenomena due to the abundance of available exact results. In this Letter, we have shown that these models can host QMBS, which lead to weak ergodicity breaking in a provably maximally chaotic system. We provide a systematic way to embed QMBS into DU circuits and highlight the contrast with the rest of the spectrum via numerical simulations.

The presented results motivate several fundamental questions with respect to QMBS, DU circuits, and chaotic quantum many-body systems, in general. Because the fact that for $d>2$ the used parametrization is not complete for DU gates, we expect that there are more DU circuit instances which can host QMBS. However, it is not certain whether the proposed embedding approach will work for DU circuits constructed with gates which lie outside the used form in Eq. \eqref{eq:new_DU_1} for $d > 2$. Furthermore, even with our parametrization, our embedding approach is probably not exhaustive. Further theoretical investigations are required to obtain a more complete picture of QMBS in DU circuits and their properties. 

\begin{acknowledgments}
We acknowledge helpful conversations Jean-Yves Desaules, Nathan Keenan, and Federica Maria Surace. J.G. and L.L. acknowledge financial support by Microsoft Ireland. S.D. and J.G. acknowledge financial support from the SFI-EPSRC joint project QuamNESS. J.G. is supported by a SFI-Royal Society University Research Fellowship and acknowledges funding from European Research Council Starting Grant ODYSSEY (Grant Agreement No. 758403). S.P. acknowledges support by the Deutsche Forschungsgemeinschaft (DFG, German Research Foundation) under Germany’s Excellence Strategy - Cluster of Excellence Matter and Light for Quantum Computing (ML4Q) EXC 2004/1-390534769.
\end{acknowledgments}


\bibliography{refs}

\clearpage

\setcounter{page}{1}
\onecolumngrid 

{\large \textbf{Supplementary Material for ``Quantum many-body scars in dual unitary circuits''}}

{\large Leonard Logari\'c, Shane Dooley, Silvia Pappalardi, John Goold}

\section{Some properties of dual-unitary gates} \label{app:U_DU_parameterisation}

For completeness, in this section we prove several properties of dual-unitary gates, including that $\hat{U}^\text{DU,1}$ and $\hat{U}^\text{DU,2}$ [Eqs. \eqref{eq:new_DU_1} and \eqref{eq:new_DU_2} in the main text] are dual-unitary (which, we note, was already shown in Ref. \cite{Poz-22}).

\begin{claim} \label{claim:local_u} Suppose that $\hat{U}_A$ and $\hat{U}_B$ are two-qudit unitary operators that are related by the local unitary transformation $\hat{U}_B = (\hat{u}_+\otimes\hat{u}_-)\hat{U}_A (\hat{v}_-\otimes\hat{v}_+)$, where $\hat{u}_\pm$, $\hat{v}_\pm$ are single qudit unitaries. Then the dual of $\hat{U}_B$ is $\tilde{U}_B = (\hat{v}_+^T \otimes \hat{u}_-) \tilde{U}_A (\hat{v}_- \otimes \hat{u}_+^T)$, where $\tilde{U}_A$ is the dual of $\hat{U}_A$. \end{claim}

\begin{proof} First, we write $\hat{U}_A = \sum_{i,j,k,l = 0}^{d-1} [U_A]_{ij}^{kl} \ket{k} \bra{i} \otimes \ket{l} \bra{j}$ in terms of the single-qudit local basis $\{ \ket{i} \}_{i=0}^{d-1}$, and we note that its dual is the operator $\tilde{U}_A = \sum_{i,j,k,l = 0}^{d-1} [U_A]_{ij}^{kl} \ket{j} \bra{i} \otimes \ket{l} \bra{k}$ obtained by exchanging the first ket and the last bra (with respect to this local basis).

\noindent Substituting the expression $\hat{U}_A = \sum_{i,j,k,l = 0}^{d-1} [U_A]_{ij}^{kl} \ket{k} \bra{i} \otimes \ket{l} \bra{j}$ into the identity $\hat{U}_B = (\hat{u}_+\otimes\hat{u}_-)\hat{U}_A (\hat{v}_-\otimes\hat{v}_+)$, we have: \begin{equation} \hat{U}_B = \sum_{i,j,k,l = 0}^{d-1} [U_A]_{ij}^{kl} ( \hat{u}_+ \ket{k} \bra{i} \hat{v}_- ) \otimes ( \hat{u}_-\ket{l} \bra{j} \hat{v}_+ ) . \label{eq:U_B} \end{equation} We also explicitly write $\hat{u}_+ = \sum_{i',k'} [u_+]_{i'}^{k'} \ket{k'}\bra{i'}$ and $\hat{v}_+ = \sum_{j',l'} [v_+]_{j'}^{l'} \ket{l'}\bra{j'}$ in the local basis and substitute into Eq. \ref{eq:U_B} to obtain: \begin{equation} \hat{U}_B = \sum_{i,j,k,l} \sum_{k',j'} [U_A]_{ij}^{kl} [u_+]_{k}^{k'} [v_+]_{j'}^j (\ket{k} \bra{i} \hat{v}_- ) \otimes ( \hat{u}_-\ket{l} \bra{j} ) .   \end{equation} In this basis, the dual of $\hat{U}_B$ is the operator obtained by exchanging the first ket and the last bra of this expression: \begin{eqnarray} \tilde{U}_B &=& \sum_{i,j,k,l} \sum_{k',j'} [U_A]_{ij}^{kl} [u_+]_{k}^{k'} [v_+]_{j'}^j (\ket{j} \bra{i} \hat{v}_- ) \otimes ( \hat{u}_-\ket{l} \bra{k} ) \nonumber \\ &=& \sum_{i,j,k,l} [U_A]_{ij}^{kl}  \left(\sum_{j'} [v_+]_{j'}^j \ket{j} \right) \bra{i} \hat{v}_-  \otimes \hat{u}_-\ket{l} \left(\sum_{k'} [u_+]_{k}^{k'} \bra{k} \right) \nonumber \\ &=& (\hat{v}_+^T \otimes \hat{u}_-) \Big\{ \sum_{i,j,k,l = 0}^{d-1} [U_A]_{ij}^{kl} \ket{j} \bra{i} \otimes \ket{l} \bra{k} \Big\} (\hat{v}_- \otimes \hat{u}_+^T) .   \end{eqnarray} Observing that the operator in curly brackets is the dual of $\hat{U}_A$ we have $\tilde{U}_B = (\hat{v}_+^T \otimes \hat{u}_-) \tilde{U}_A (\hat{v}_- \otimes \hat{u}_+^T)$. \qed \end{proof}

\begin{corollary} If $\hat{U}_A$ is dual-unitary then so is $\hat{U}_B = (\hat{u}_+\otimes\hat{u}_-)\hat{U}_A (\hat{v}_-\otimes\hat{v}_+)$. \end{corollary}

\begin{proof} The transpose of any unitary operator is unitary. So, if $\tilde{U}_A$ is unitary so is $\tilde{U}_B = (\hat{v}_+^T \otimes \hat{u}_-) \tilde{U}_A (\hat{v}_- \otimes \hat{u}_+^T)$. \qed \end{proof}

\begin{claim} \label{claim:SUS} If $\hat{U}$ is dual-unitary then $\hat{S}\hat{U}\hat{S}$ is also dual-unitary. Its dual is $\tilde{U}^T$, where $\tilde{U}$ is the dual of $\hat{U}$.  \end{claim}

\begin{proof}
Write $\hat{U} = \sum_{i,j,k,l = 0}^{d-1} U_{ij}^{kl} \ket{k} \bra{i} \otimes \ket{l} \bra{j}$ and its dual $\tilde{U} = \sum_{i,j,k,l = 0}^{d-1} U_{ij}^{kl} \ket{j} \bra{i} \otimes \ket{l} \bra{k}$. Then we have: \begin{equation} \hat{S}\hat{U}\hat{S} = \sum_{ijkl} U_{ij}^{kl} \ket{l} \bra{j} \otimes \ket{k} \bra{i} . \end{equation} Its dual is obtained by exchanging the first ket and the last bra in this basis: \begin{eqnarray} \widetilde{SUS} &=& \sum_{ijkl} U_{ij}^{kl} \ket{i} \bra{j} \otimes \ket{k} \bra{l} \nonumber \\ &=& \tilde{U}^T . \end{eqnarray} If $\tilde{U}$ is unitary then so is $\tilde{U}^T$, and therefore $\widetilde{SUS}$ is also unitary. \qed
\end{proof}

\begin{claim} \label{claim:SV} The unitary operator $\hat{U} = \hat{S}\hat{V}$ is dual-unitary, and its dual is $\tilde{U} = \hat{U} = \hat{S}\hat{V}$. \end{claim}

\begin{proof} We write the swap gate and the entangling gate explicitly in terms of the local basis, $\hat{S} = \sum_{i,j=0}^{d-1} \ket{i} \bra{j} \otimes \ket{j}\bra{i}$ and $\hat{V} = \exp\{ i\sum_{j=0}^{d-1} \hat{h}^{(j)} \otimes \ket{j} \bra{j} \}$. This gives: \begin{equation} \hat{U} = \hat{S}\hat{V} = \sum_{i,j} \ket{i} \bra{j} e^{i\hat{h}^{(i)}} \otimes \ket{j}\bra{i} . \end{equation} Now, in this basis its dual operator is obtained by exchanging the first ket and the last bra: \begin{eqnarray} \tilde{U} &=& \sum_{i,j} \ket{i} \bra{j} e^{i\hat{h}^{(i)}} \otimes \ket{j}\bra{i} \nonumber \\ &=& \hat{U} . \qed \end{eqnarray} \end{proof}

\begin{corollary} \label{cor:U_DU_1} The unitary operator $\hat{U}^\text{DU,1} = (\hat{u}_+ \otimes \hat{u}_-) \hat{S}\hat{V} (\hat{v}_- \otimes \hat{v}_+)$ is dual-unitary and its dual is $\tilde{U}^\text{DU,1} = (\hat{v}_+^T \otimes \hat{u}_-) \hat{S}\hat{V} (\hat{v}_- \otimes \hat{u}_+^T)$. \end{corollary}

\begin{proof} This follows from Claim \ref{claim:local_u} and Claim \ref{claim:SV}. \end{proof}

\begin{corollary} \label{cor:U_DU_2} The unitary operator $\hat{U}^\text{DU,2} = \hat{S} \hat{U}^{DU,1} \hat{S}$ is dual-unitary. Its dual is $(\tilde{U}^{DU,1})^T = (\hat{v}_-^T \otimes \hat{u}_+)\hat{V}^T \hat{S} (\hat{v}_+ \otimes \hat{u}_-^T)$.
\end{corollary}

\begin{proof} This follows from Corollary \ref{cor:U_DU_1} and Claim \ref{claim:SUS} \end{proof}

\begin{claim} 
The parameterisation in Eq. \eqref{eq:new_DU_1} includes all DU gates in the case of local Hilbert space dimension $d = 2$.
\end{claim}

\begin{proof}
    For $d = 2$, any DU gate can be written in the form \cite{Ber-19a}:
    \begin{equation}
        \hat{U}^{\text{DU}} = e^{i \phi '} (\hat{u}^+ \otimes \hat{u}^-) e^{-i \frac{\pi}{4}(\hat{\sigma}_x \otimes \hat{\sigma}_x + \hat{\sigma}_y \otimes \hat{\sigma}_y + J \hat{\sigma}_z \otimes \hat{\sigma}_z)} (\hat{w}^- \otimes \hat{w}^+).
    \end{equation}
    Furthermore, it can also be rewritten as \cite{Suz-22}:
    \begin{equation} \label{eq:complete_DU_d_2}
        \hat{U}^{\text{DU}} = e^{i \phi} (\hat{u}^+ \otimes \hat{u}^-) \hat{S} (\hat{C}_z) ^{\alpha} (\hat{v}^- \otimes \hat{v}^+),
    \end{equation}
    where the parameters and single-qubit unitaries are related by: $\phi = \phi ' + \frac{\pi}{4} \alpha$, $\alpha = j - 1$, $\hat{v}^- = e^{-i \frac{\pi}{4} \alpha \hat{\sigma}_z}$, $\hat{v}^+ = e^{-i \frac{\pi}{4} \alpha \hat{\sigma}_z}$. The $\hat{C}_z$ operator represents the controlled-Z gate, which can be written as a diagonal matrix in the computational basis: $\hat{C}_z = \text{diag}(1, 1, 1, -1)$. Therefore we can rewrite $(\hat{C}_z)^{\alpha}$ as: 
    \begin{equation}
        (\hat{C}_z)^{\alpha} = \text{exp}[\alpha \text{ln}(\hat{C}_z)] = \text{exp}[ \text{diag}(0, 0, 0, \alpha \pi i)] = \text{exp}[\text{diag} (0, \alpha \pi i) \otimes \ket{1} \bra{1}].
    \end{equation}
    Therefore it is easy to see that the complete parameterisation from \eqref{eq:complete_DU_d_2} is included in the form in \eqref{eq:new_DU_1}, with $\hat{h}^{(0)}_{i,j} = 0$ and $\hat{h}^{(1)} = \text{diag}(0, \alpha \pi i)$.
\end{proof}

\section{Ergodic Properties of Presented Models}\label{appendix:ergodic_properties}

The ergodic properties of a quantum many-body system can be characterized by the behaviour of its dynamical correlation functions \cite{Arn-89a}. For DU circuits, in particular, it has been shown that the dynamical correlation functions can be calculated exactly \cite{Ber-19a}, and are directly related to the spectral properties of a completely-positive and trace-preserving map $\mathcal{M}_{\nu}$, where $\nu = \pm$ denotes the direction along which we are computing the correlation function. For our circuit setup, with DU gates $\hat{U}_e^\text{DU}$ making up the even layer and DU gates $\hat{U}_o^\text{DU}$ making up the odd layer, the map can be written as $\mathcal{M}_{\nu} = \mathcal{M}_{\nu, o} \mathcal{M}_{\nu, e}$, with:
\begin{align}
    \mathcal{M}_{+,e/o} (\hat{o}) \equiv \frac{1}{d} \Tr_1 [(\hat{U}^{\text{DU}}_{e/o})^\dagger (\hat{o} \otimes \hat{\mathbb{I}}) \hat{U}^\text{DU}_{e/o}], \\
    \mathcal{M}_{-,e/o} (\hat{o}) \equiv \frac{1}{d} \Tr_2 [(\hat{U}^{\text{DU}}_{e/o})^\dagger (\hat{\mathbb{I}} \otimes \hat{o}) \hat{U}^\text{DU}_{e/o}],
\end{align}
where $\hat{o}$ is an arbitrary local observable and $\Tr_{1/2}$ is a partial trace over the left/right qudit \cite{Fri-21a}. 

The mathematical properties of $\mathcal{M}_\nu$ guarantee that all of its eigenvalues lie on or within a unit circle, centered at the origin of the complex plane. The single qudit identity operator $\hat{o} = \hat{\mathbb{I}}$ is always an eigenoperator of $\mathcal{M}_{\nu}$ with unit eigenvalue, i.e., $\mathcal{M}_\pm (\hat{\mathbb{I}}) = \hat{\mathbb{I}}$. If there is one or more additional eigenoperator of $\mathcal{M}_\pm$ with unit eigenvalue then the circuit is \emph{non-ergodic}, since these additional eigenoperators correspond to non-decaying dynamical correlation functions. On the other hand, if the single-qudit identity operator is the only eigenvector of $\mathcal{M}_\pm$ with unit eigenvalue then the circuit is \emph{ergodic} and all time-averaged dynamical correlation functions approach their thermal values at large times. If, in addition, there are no other eigenvalues of $\mathcal{M}_\pm$ with unit modulus, then the circuit is also \emph{mixing}.

\begin{figure}
    \includegraphics[width=\textwidth]{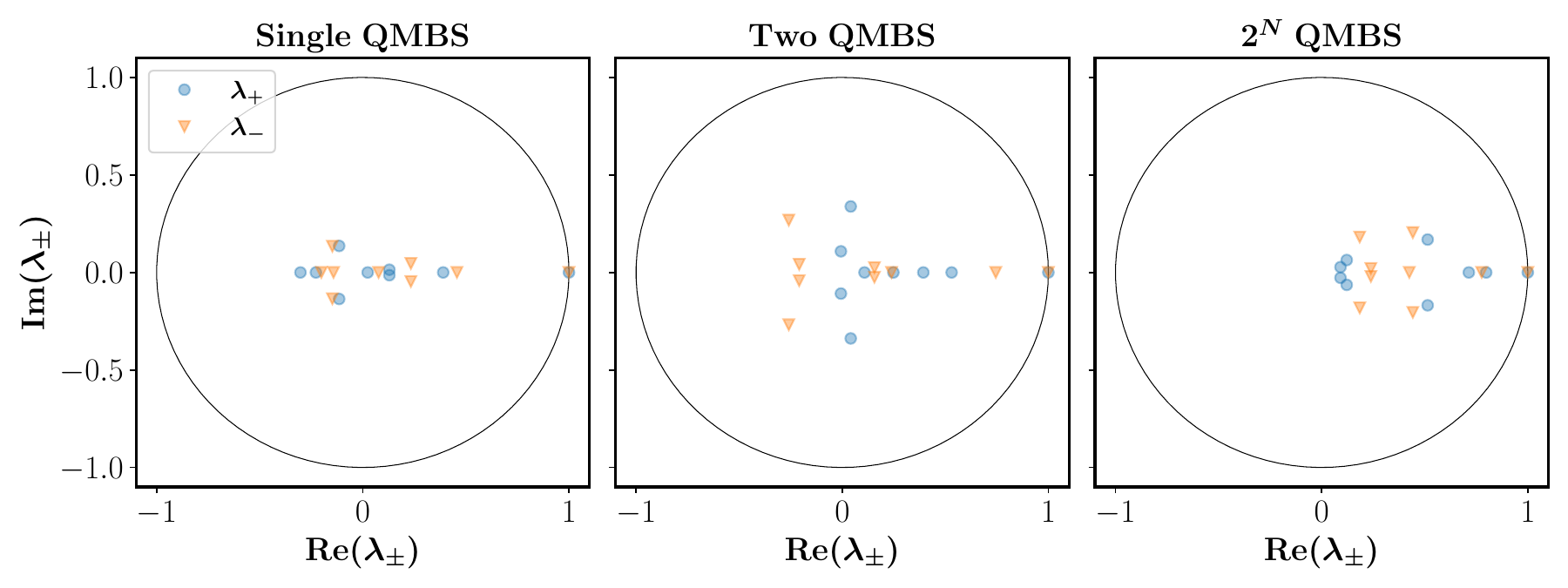}
    \caption{Eigenvalues of the $\mathcal{M}_{\pm}$ map for the three models presented in this work. For all of the models we can observe exactly one eigenvalue of $\lambda_{\pm} = 1.0$, corresponding to the identity operator $\hat{o} = \hat{\mathds{1}}$. All of the other eigenvalues of $\mathcal{M}_{\pm}$ satisfy $|\lambda_{\pm}| < 1$, thus demonstrating that the presented models fall within the ergodic and mixing class according to the classification presented in \cite{Ber-19a}.}
    \label{fig:M_map}
\end{figure}


For the DU circuit setup in this work we have chosen the even layer DU gates to be of the form $\hat{U}_e^\text{DU} = \hat{U}^{DU, 1}$ (c.f. Eq. \ref{eq:new_DU_1}) and the odd layer DU gates to be of the form $\hat{U}_o^\text{DU} = \hat{U}^{DU, 2}$ (c.f. Eq. \ref{eq:new_DU_2}). With this choice we obtain a map $\mathcal{M}_{\pm}$ that has a single eigenvalue with unit modulus, corresponding to the identity eigenoperator, and therefore the DU circuit will display ergodic and mixing behaviour according to the classification described above. This is shown numerically in Fig. \ref{fig:M_map} where we plot the eigenvalues $\lambda_\pm$ of $\mathcal{M}_\pm$ for each of our three examples (example A and example B in the main text, as well as the two QMBS example in section \ref{appendix:two_qmbs}). [We have observed that the circuit is also ergodic and mixing if we choose $\hat{U}_e^\text{DU} = \hat{U}^{DU, 2}$ and $\hat{U}_o^\text{DU} = \hat{U}^{DU, 1}$.] 

However, it is also possible to embed QMBS in DU circuits that are \emph{non-ergodic}. This can be achieved by choosing even and odd layer gates to be the same [i.e., $\hat{U}_e^\text{DU} = \hat{U}_o^\text{DU} = \hat{U}^{DU, 1}$ or $\hat{U}_e^\text{DU} = \hat{U}_o^\text{DU} = \hat{U}^{DU, 2}$]. In this case the $\mathcal{M}_{\pm}$ maps will have 2 unit eigenvalues, and thus the DU circuit will be non-ergodic by the classification described above.

\section{Explicit construction} \label{appendix:numerical_demonstrations}

We here explicitly describe a construction of dual-unitary gates (defined by the $d \times d$ Hermitian matrices $\{ \hat{f}^\pm, \hat{g}^\pm, \hat{h}^{(j)} \}$) satisfying conditions \eqref{eq:C1}--\eqref{eq:C3}. We restrict our considerations to projectors of the form \begin{equation} \hat{P}_{n,n+1} = \hat{\mathbb{I}}_{n,n+1} - \sum_{x\in\mathcal{X}} \ket{x}\bra{x}_{n,n+1} , \label{eq:X} \end{equation} where $\mathcal{X}$ is some subset of products of computational basis states, i.e., $\ket{x}_{n,n+1} = \ket{x'}_n\otimes\ket{x''}_{n+1}$ where $\ket{x'}, \ket{x''} \in \{ \ket{i} \}_{i=0}^{d-1}$. For projectors of this form, conditions \eqref{eq:C1} and \eqref{eq:C2} can be satisfied by choosing the single-qudit Hermitian operators $\hat{f}^\pm$ and $\hat{g}^\pm$ such that $\hat{f}^+ \ket{x'} = \hat{g}^- \ket{x'} = 0$ and $\hat{f}^- \ket{x''} = \hat{g}^+ \ket{x''} = 0$, for all $\ket{x}=\ket{x',x''}\in\mathcal{X}$. This corresponds to setting some rows and columns of these matrices to zero (in the computational basis $\{ \ket{i} \}_{i=0}^{d-1}$). All of the other matrix elements can be chosen arbitrarily. It still remains to constuct an entangling gate $\hat{V}_{n,n+1} = \exp\{ \sum_{j=0}^{d-1} \hat{h}^{(j)} \otimes \ket{j}\bra{j} \}$ that satisfies condition \eqref{eq:C3}.
First, consider the two-qudit operator $\hat{\mathsf{H}}_{n,n+1} = \sum_{j=0}^{d-1} \hat{h}^{\prime (j)} \otimes \ket{j}\bra{j}$ where $\hat{h}^{\prime(j)}$ are arbitrary single-qudit Hermitian operators. This operator does not necessarily satisfy \eqref{eq:C3}. However, the projected operator $\hat{H}_{n,n+1} = \hat{P}_{n,n+1} \hat{\mathsf{H}}_{n,n+1} \hat{P}_{n,n+1}$ is still of the form $\hat{H}_{n,n+1} = \sum_{j=0}^{d-1} \hat{h}^{(j)} \otimes \ket{j}\bra{j}$ and is guaranteed to satisfy the condition \eqref{eq:C3}, since we have $\hat{P}_{n,n+1} \hat{H}_{n,n+1} \hat{P}_{n,n+1} = \hat{H}_{n,n+1}$. 

In all of  our numerical examples we construct a set of Hermitian generators $\{\hat{f}^{\pm}, \hat{g}^{\pm}, \hat{h}^{(j)}\}$ by first generating random matrices with elements $m_{j,k} = \alpha_{j,k} + i\beta_{j, k}$, where $\alpha_{j,k}, \beta_{j, k}$ are sampled from a uniform distribution. Random Hermitian matrices are then obtained as, e.g., $\hat{h} = \hat{m} + \hat{m}^{\dagger}$. Finally, we impose the conditions \eqref{eq:C1}, \eqref{eq:C2}, \eqref{eq:C3} in the manner described above. This procedure is used in the following three examples:

\begin{enumerate}[label=(\Alph*)] 
\item \label{proj: one scar} $\mathcal{X} = \{ \ket{0,0} \}$, with the target space $\mathcal{T} = \{ \ket{0}^{\otimes N} \}$ consisting of a single QMBS. The corresponding results can be found in \textbf{\emph{example A}} in the main text, and in Fig. \ref{fig:QMBS_in_DUC}.

\item \label{proj: exp scars} $\mathcal{X} = \{ \ket{0,0}, \ket{d-1,d-1}, \ket{0,d-1}, \ket{d-1,0} \}$, in this case the target space is $\mathcal{T} = \text{span}\{ \ket{0}, \ket{d-1} \}^{\otimes N}$, which has dimension $\text{dim}(\mathcal{T}) = 2^N$. This gives an exponentially growing number of QMBS, though still an exponentially small fraction $(2/d)^N$ of all eigenstates for $d>2$. The corresponding results can be found in \textbf{\emph{example B}} in the main text, and in Fig. \ref{fig:QMBS_in_DUC}. We note that one can in principle also embed a number of scars which scales exponentially, but with a different base $k$ (with $k < d$). One can achieve this by setting $\mathcal{X} = \{ \ket{x} \otimes \ket{y}: x, y = 0, 1, ..., k-1 \}$. This would result in the target space $\mathcal{T} = \text{span}\{\ket{0}, \ket{1}, ..., \ket{k - 1} \}^{\otimes N}$, ie. the number of scars would be $k^N$ and still an exponentially small fraction $(k/d)^N$ of all eigenstates.

\item \label{proj: two scars} $\mathcal{X} = \{ \ket{0,0}, \ket{d-1,d-1} \}$, with the corresponding target space $\mathcal{T} = \text{span}\{ \ket{0}^{\otimes N} , \ket{d-1}^{\otimes N} \}$, containing two QMBS. The results are presented in section \ref{appendix:two_qmbs}.
\end{enumerate} 

\section{Breaking QMBS eigenphase degeneracy}
\label{app:QMBS_degeneracy}

In examples with multiple QMBS it is possible that they can have a degeneracy in the eigenphase $\varphi_\alpha$ (where $\hat{\mathbb{U}} \ket{\varphi_\alpha} = e^{i\varphi_\alpha} \ket{\varphi_\alpha}$). To determine the different possible values of the eigenphase for QMBS we observe that for QMBS we have $\hat{\mathbb{U}}\ket{\varphi_\alpha} = \hat{\mathbb{S}} \ket{\varphi_\alpha} = e^{i\varphi_\alpha} \ket{\varphi_\alpha}$ and that $\hat{\mathbb{S}}^{N/2} = \hat{\mathbb{I}}$. Hence:

\begin{eqnarray}
    \hat{\mathbb{S}}^{N / 2} \ket{\varphi _{\alpha}} &=& \exp\{i N \varphi _{\alpha} /2 \} \ket{\varphi _{\alpha}} = \ket{\varphi _{\alpha}}  \\
    &\implies& e^{i N \varphi _{\alpha} /2} = 1 \\ &\implies& \varphi _{\alpha} = \frac{4 \pi k}{N}, 
\end{eqnarray}
where $k \in \{ -\lfloor\frac{N}{4}\rfloor, -\lfloor\frac{N}{4}\rfloor + 1, \hdots, \lfloor\frac{N}{4}\rfloor - 1  \}$. Since there are only $N/2$ distinct values of the eigenphase there must be degeneracies when there are more than $N/2$ QMBS, for instance, in examples where we have an exponenential number $2^N$ of QMBS.

However, some degeneracies in the QMBS in DU circuits can be broken by adding an additional unitary layer $\hat{\mathbb{U}}'$ at each timestep of the circuit with the properties $[ \hat{\mathbb{S}}, \hat{\mathbb{U}}'] = 0 $ and $[\hat{\mathbb{P}}_{n,n+1}, \hat{\mathbb{U}}'] = 0$ for all $n$. These properties ensure that the projectors and the swap layers share a set of simultaneous eigenstates with $\hat{\mathbb{U}}'$, including the QMBS in $\mathcal{K}$. If $\hat{\mathbb{U}}'$ can be constructed as a product of one-qudit gates then these gates can be absorbed into the even or odd layer unitaries $\hat{\mathbb{U}}_{e/o}$ without disturbing the brickwork structure of the circuit and without impacting dual unitarity. This approach to break the QMBS eigenphase degeneracy is used in the example in section \ref{appendix:two_qmbs}, where we embed two QMBS in a DU circuit.

\section{An example with two QMBS} \label{appendix:two_qmbs}

In order to embed two QMBS into a DU circuit, we choose the following set of projectors:
\begin{equation}
    \hat{P}_{n, n+1} = \hat{\mathds{1}}_{n, n+1} - \ket{0}\bra{0}_n \otimes \ket{0}\bra{0}_{n + 1} - \ket{d-1}\bra{d-1}_n \otimes \ket{d-1}\bra{d-1}_{n + 1},
\end{equation}
which will lead to the common kernel $\mathcal{K} = \{\ket{0}^{\otimes N}, \ket{d-1}^{\otimes N} \}$. Since any state of the form $\ket{i}^{\otimes N}$ is invariant under the application of a layer of swap gates $\hat{\mathbb{S}}_{e/o}$, the target space will be the same as the kernel, $\mathcal{T} = \mathcal{K}$. In order to construct the DU gate generators $\{\hat{f}^\pm , \hat{g}^\pm, \hat{h}^{(j)} \}$ we follow the procedure described in section \ref{appendix:numerical_demonstrations}, guaranteeing that they will obey the required conditions \eqref{eq:C1}, ~\eqref{eq:C2}, ~\eqref{eq:C3}. Once the generators are constructed, we numerically calculate the same quantities as for examples A and B in the main text. The result is a circuit with the two QMBS $\hat{\mathbb{U}} \ket{0}^{\otimes N} = \ket{0}^{\otimes N}$ and $\hat{\mathbb{U}} \ket{d-1}^{\otimes N} = \ket{d-1}^{\otimes N}$ that have the same eigenphase $\varphi = 0$.

However, we also break the QMBS eigenphase degeneracy by adding an additional layer of single qudit unitaries $\hat{\mathbb{U}}\to \hat{\mathbb{U}}\hat{\mathbb{U}}'$, where $\hat{\mathbb{U}}' = \hat{u}_\phi^{\otimes N}$ and $\hat{u}_{\phi} = \sum_{j=0}^{d-1} \exp \{-i\phi \frac{2j-(d-1)}{d-1} \} |j\rangle\langle j |$ are generalised Pauli $Z$-rotations (see Appendix \ref{app:QMBS_degeneracy} above). This additional layer of single-qudit unitaries will add a phase of $e^{iN \phi}$ to the $\ket{0}^{N}$ state, whereas the other QMBS state, $\ket{d-1}^N$, will acquire an eigenphase of $e^{-iN \phi}$. In our simulations we set $\phi = 0.01$. 

The numerical results are shown in Fig. \ref{fig:two_QMBS_in_DUC}. The entanglement growth (Fig \ref{fig:two_QMBS_in_DUC} - centre) and fidelity evolution (Fig. \ref{fig:two_QMBS_in_DUC} - right) are quantitatively very close to the ones obtained for the single QMBS case (\ref{fig:QMBS_in_DUC} - b, c). The only major difference can be observed in the plot of entanglement entropies and eigenphases of the individual eigenstates (Fig. \ref{fig:two_QMBS_in_DUC} - left), where we now observe two QMBS at $0$ entanglement entropies, both having a non-zero eigenphases due to the choice of phase degeneracy breaking single-qudit unitaries.



\begin{figure}
    \includegraphics[width=\textwidth]{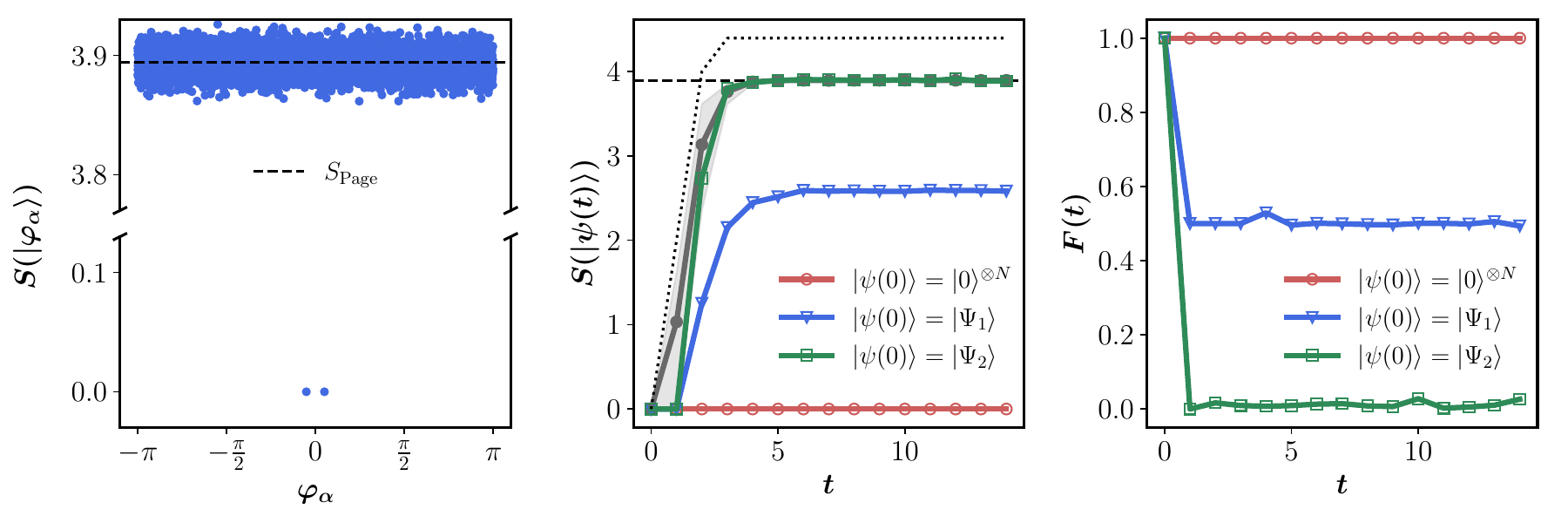}
    \caption{An example illustrating out construction with two QMBS. (Left) The entangement entropy and eigenphases of individual eigenstates of the Floquet operator. Most of the spectrum has an entanglement entropy close to the Page value, while the QMBS states, $\{\ket{0}^{\otimes N}, \ket{d-1}^{\otimes N} \}$ are at $0$ entanglement entropy and have non-zero phases, due to the chosen single-qudit unitaries. (Centre) Entanglement entropy growth for the set of initial states as used iA two-qudit unitary:n Fig. \ref{fig:QMBS_in_DUC}, and we observe quantitatively similar results as in example A. (Right) Time evolution of fidelity with respect to the initial state.  [Parameters: $N=8$, $d=3$]}
    \label{fig:two_QMBS_in_DUC}
\end{figure}

\section{Dual quantum many-body scars}
\label{app:dual_QMBS}

The Floquet unitary operator $\hat{\mathbb{U}}$ propagates the $N$-qudit system forward one timestep, while the $\tau$'th power $\hat{\mathbb{U}}^\tau$ propagates the system $\tau$ steps forward in time. However, since our circuit is DU, we can also define a \emph{dual} Floquet unitary operator $\tilde{\mathbb{U}}$, which propagates a $\tau$-depth circuit one step forward in the spatial direction. This raises the question: if the Floquet unitary $\hat{\mathbb{U}}$ hosts QMBS, does the dual Floquet unitary $\tilde{\mathbb{U}}$ also host QMBS? In this section we show numerically that for example A and example B in the main text the QMBS in $\hat{\mathbb{U}}$ have counterpart \emph{dual} QMBS in the dual operator $\tilde{\mathbb{U}}$.

We construct the dual Floquet unitary as the product of an odd and even layer of dual two-qudit gates, $\tilde{\mathbb{U}} = \tilde{\mathbb{U}}_o \tilde{\mathbb{U}}_e$ where $\tilde{\mathbb{U}}_e = \bigotimes_{t=0}^{\tau - 1} \tilde{U}^\text{DU,1}_{t,t+\frac{1}{2}}$ and $\tilde{\mathbb{U}}_o = \bigotimes_{t=0}^{\tau - 1} \tilde{U}^\text{DU,2}_{t+\frac{1}{2}, t+1}$, and we assume periodic boundary conditions (in time). Here, $\tilde{U}^\text{DU, 1}_{t, t + \frac{1}{2}}$ is the dual of the two-qudit gate $\hat{U}^\text{DU, 1}$ in the even layer at timestep $t$, while $\tilde{U}^\text{DU,2}_{t+\frac{1}{2},t+1}$ is the dual of the two-qudit gate $\hat{U}^\text{DU,2}$ in the odd layer at timestep $t$ (a single timestep includes both an even and an odd layer). In Fig. \ref{fig:dual_QMBS} we plot the bipartite entanglement entropy of the eigenstates $\tilde{\mathbb{U}} \ket{\tilde{\varphi}_\alpha} = e^{i \tilde{\varphi}} \ket{\tilde{\varphi}_\alpha}$. We observe that the bulk of the eigenstates are thermal, while there is a single dual-QMBS (for example A) or $2^{2\tau}$ dual-QMBS (for example B) with zero or low entanglement.

\begin{figure}
\includegraphics[width=0.45\textwidth]{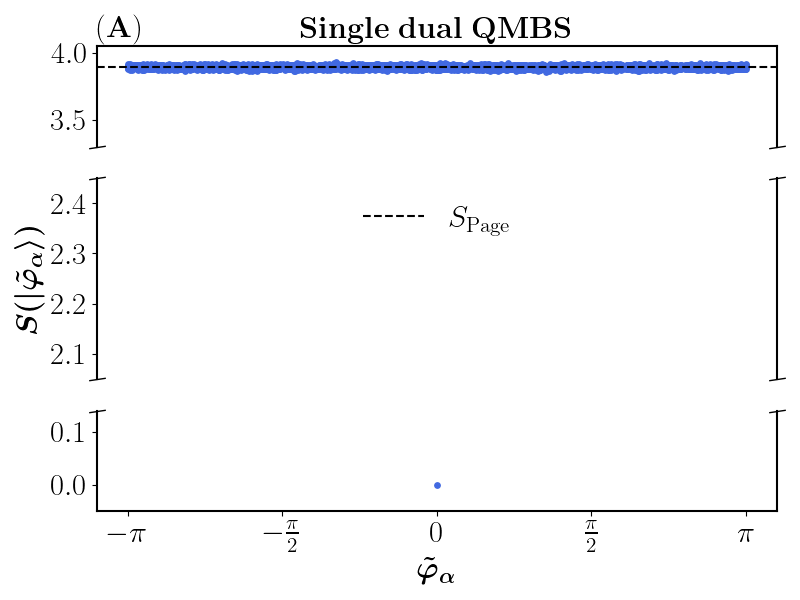}
\includegraphics[width=0.45\textwidth]{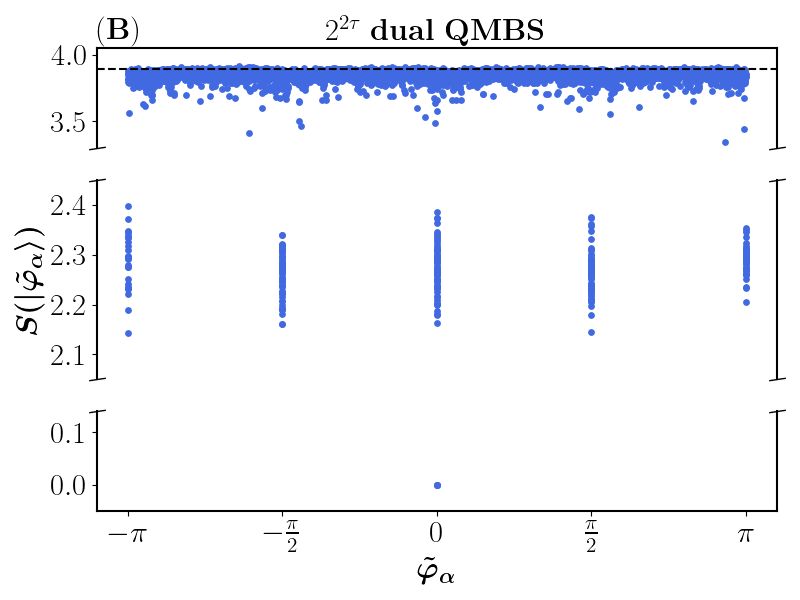}
    \caption{The bipartite entanglement entropy of the eigenstates $\ket{\tilde{\varphi}_\alpha}$ of the dual Floquet unitary $\tilde{\mathbb{U}}$. Both example A (left) and example B (right) from the main text reveal \emph{dual}-QMBS -- counterparts to the standard QMBS in each example.  [Parameters: $\tau=4$, $d=3$]}
    \label{fig:dual_QMBS}
\end{figure}

\section{Upper bound for QMBS entanglement}
\label{app:QMBS_entanglement_bound}

One of the characteristics of QMBS is that they have a sub-volume law scaling of bipartite entanglement entropy \cite{Ser-21, Mou-22a, pap-2022}, in contrast to thermal eigenstates, which have a volume-law scaling. In this section we show that the QMBS in DUC by our construction have an entanglement entropy that is upper bounded by the logarithm of the volume, $S \leq \log (N/2)$ \cite{Sur-PC}.

For any state $\ket{\psi} \in \mathcal{T}$ in the QMBS subspace it follows from our construction that $\hat{\mathbb{U}}\ket{\psi} = \hat{\mathbb{S}} \ket{\psi}$. The QMBS are therefore obtained by diagonalising $\hat{\mathbb{S}}$ in the subspace $\mathcal{T}$. Any of the QMBS can be expressed as a superposition of product states. However, since $\hat{\mathbb{S}}^{N/2} = \hat{\mathbb{I}}$, they can be expressed as superpositions of at most $N/2$ orthonormal product states. It is straightforward to show that any such state has a bipartite entanglement that is bounded by $S(\ket{\psi}) \leq \log (N/2)$. 

We note that this bound may be violated in some of our numerical results due to an eigenphase degeneracy in the QMBS, resulting in the QMBS being formed in the numerics as random superpositions of degenerate orthonormal states that obey the entanglement bound. However, in for such degenerate QMBS is always possible to choose them so that they obey the entanglement bound.

\section{ETH and dynamics of local observables}\label{appendix:local_obs}

By the eigenstate thermalisation hypothesis (ETH) we expect thermalising systems to have eigenstate expectation values that are close to the thermal expectation values (for local observables). If all eigenstates are thermal then the corresponding observables will thermalise, starting from a non-equilibrium state. However, non-thermal eigenstates, like QMBS can prevent thermalisation. In our circuit model, since energy is not conserved, the thermal expectation values are with respect to the infinite temperature state $\hat{\rho}_{T = \infty} = \hat{\mathbb{I}}/d^N$.

We check these features in our DU circuit with QMBS for the local observable \begin{equation}
    \hat{S}^z_\text{tot} = \sum_{n=0}^{N-1} \hat{S}^z_n, \quad \hat{S}_n^z = \sum_{i=0}^{d-1} \frac{2i - (d-1)}{d-1} \ket{i}\bra{i}_n,
\end{equation}
where $\hat{S}_n^{z}$ is the local magnetisation at site $n$. In Fig. \ref{fig:QMBS_in_DUC_exp_vals}, top row, we can clearly see that the QMBS in the three example models are well separated from the rest of the spectrum, which is concentrated around $S_{tot}^z \approx 0$, the expectation value for the infinite temperature thermal state. These QMBS can prevent thermalisation of the total magnetisation $\hat{S}^z_\text{tot}$. To show this, we also numerically calculate the dynamics of the expectation value of $\hat{S}_{tot}^{z} / N$ for the set of initial states that were previously used to probe entanglement growth and fidelity evolution. 

As expected, the $\ket{0}^{\otimes N}$ state, which is a QMBS in all three models, gives a failure to thermalise in any of the examples. 

The $\ket{\Psi _1} = \ket{0}^{\otimes N-1}\otimes (\ket{0} + \ket{d-1})/\sqrt{2}$ state is an equal superposition of the QMBS $\ket{0}^{\otimes N}$, and the state $\ket{0}^{\otimes N-1}\ket{d-1}$ which is in the QMBS subspace only for the model with exponentially many QMBS. Hence, for the single- and two-QMBS examples we observe the magnetisation density approaching $\langle \hat{S}_\text{tot}^z \rangle/N \approx -0.5$ (since the non-QMBS component thermalises but the QMBS component does not). In the model with exponentially many QMBS, the state $\ket{\Psi _1}$ is in the QMBS subspace. Although the state evolves in a non-trivial way [see Fig. \ref{fig:QMBS_in_DUC}(f) in the main text] the dynamics by the swap circuit $\hat{\mathcal{S}}$ preserves the total magnetization, as shown in the blue line in the bottom right plot of Fig. \ref{fig:QMBS_in_DUC_exp_vals}.

The $\ket{\Psi _2} = \ket{0,0,d-1,d-1}^{\otimes N/4-1}\otimes\ket{0,0,d-1}\otimes (\ket{d-1} + \ket{1})/\sqrt{2}$ initial state has an expectation value of $\langle \hat{S}_\text{tot}^z \rangle/N = 0$ for the total magnetisation, but is not in the QMBS subspace for the single- and two-QMBS examples. Therefore we initially see some oscillations at early times, with the expectation value approaching the thermal one for $t \sim \mathcal{O}(10)$. However, for the exponential QMBS example, the initial state $\ket{\Psi _2}$ is entirely in the QMBS subspace, and therefore undergoes dynamics by the swap circuit $\hat{\mathcal{S}}$ which conserves the total magnetization.

\begin{figure}
    \includegraphics[width=\textwidth]{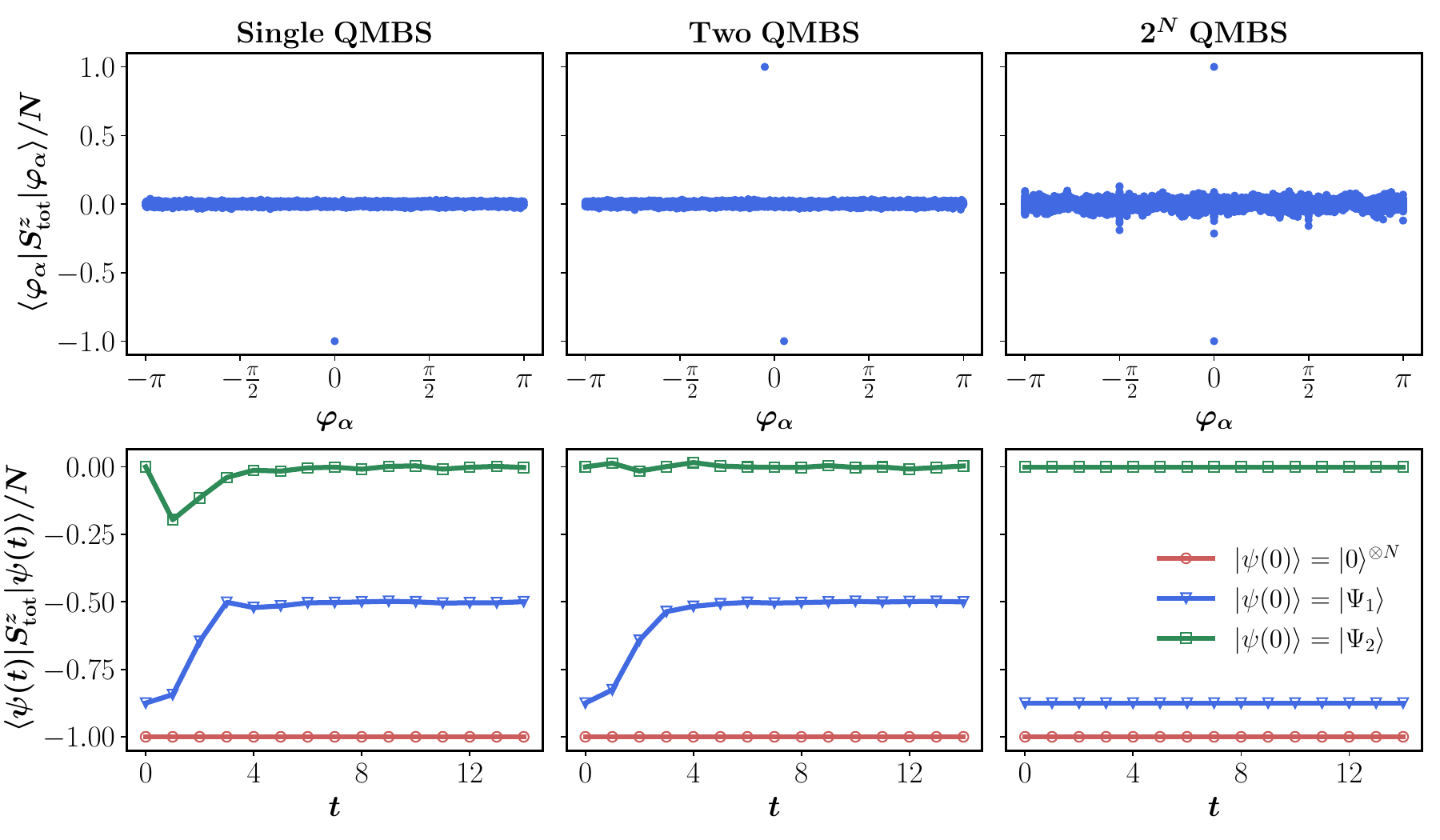}
    \caption{QMBS in DU circuits, shown through the observable $\hat{S}^z_\text{tot} = \sum_{n=0}^{N-1} \hat{S}^z_n$, where $\hat{S}_n^z = \sum_{i=0}^{d-1} (2i-d)\ket{i}\bra{i}_n$ is the local magnetization at site $n$. Top row: the expectation values of eigenstates of the Floquet operator $\hat{\mathbb{U}}$ for the observable $\hat{S}^z_\text{tot}/N$. Most eigenstates are concentrated around the infinite-temperature thermal value $\langle \hat{S}_\text{tot}^z \rangle_\text{th} = \Tr [\hat{S}_\text{tot}^z \hat{\mathbb{I}}]/d^N = 0$, though QMBS are clearly visible in our three examples. Bottom row: dynamics of the expectation value of the total magnetization, starting from different separable states. The initial state $\ket{\psi(0)} = \ket{0}^{\otimes N}$ (red line) is a QMBS in each example, and therefore gives no change in the total magnetization $\hat{S}_\text{tot}^z$. Other initial states are $\ket{\Psi_1} = \ket{0}^{\otimes N-1}\otimes (\ket{0} + \ket{d-1})/\sqrt{2}$ (blue) and $\ket{\Psi_2} = \ket{0,0,d-1,d-1}^{\otimes N/4-1}\otimes\ket{0,0,d-1}\otimes (\ket{d-1} + \ket{1})/\sqrt{2}$ green. [Parameters: $N=8$, $d=3$]}
    \label{fig:QMBS_in_DUC_exp_vals}
\end{figure}

\end{document}